\documentclass{ifacconf}

\usepackage{graphicx}
\usepackage{natbib}
\usepackage{amsmath,amssymb,bm} 
\usepackage{float}
\usepackage{subcaption}
\usepackage{xcolor}

\makeatletter
\newenvironment{proof}[1][Proof]{%
  \par\noindent\textit{#1. }\rmfamily}{\hfill$\square$\par}
\makeatother

\newtheorem{definition}{Definition}
\newtheorem{proposition}{Proposition}

\newtheorem{lemma}{Lemma}

\definecolor{darkgreen}{rgb}{0.0, 0.7, 0.0}

\newcommand{\eG}{g}
\newcommand{\ug}{s}

\begin{document}
\begin{frontmatter}

\title{Safety-Critical Control on Lie Groups Using Energy-Augmented Zeroing Control Barrier Functions} 

%

\author[Unibos]{Alessandro Letti} 
\author[RaM]{Riccardo Zanella} 
\author[Unibo]{Alessandro Macchelli} 
\author[RaM]{Federico Califano}

\address[Unibos]{Student at University of Bologna. email: alessandro.letti@studio.unibo.it}

\address[RaM]{Robotics and Mechatronics (RaM), University of Twente, 
Drienerlolaan 5, 7522 NB Enschede, The Netherlands, 
(e-mail: f.califano@utwente.nl, r.zanella@utwente.nl)}

\address[Unibo]{Department of Electrical, Electronic and Information Engineering (DEI),
University of Bologna, Italy. email: a.macchelli@unibo.it.}

\begin{abstract}
We study safety-critical control on fully actuated mechanical systems by means of Zeroing Control Barrier Functions (ZCBFs) defined on Lie Groups. 
Specifically, we introduce and theoretically validate two classes of ZCBFs. The first enforces kinematic constraints, suitable for implementing obstacle avoidance algorithms. The second enforces kinetic energy limits along prescribed inertial-frame translational and rotational directions, relevant for ensuring safe physical interaction.  
Numerical simulations involving slit traversal and safe landing scenarios are presented to validate the effectiveness and versatility of the proposed methodology.

\end{abstract}


\end{frontmatter}


\section{Introduction}
Geometric control provides a coordinate-free framework for addressing stabilization and tracking problems in mechanical systems whose configurations evolve on manifolds \citep{bullo2019geometric}. The case of fully actuated mechanical systems is well studied, and intrinsic controllers built from Riemannian geometric principles have been developed, yielding globally valid and singularity-free control laws \citep{bullo1999tracking}. A case of interest is concretized when the configuration manifold is a Lie Group, for example the group of rotations $SO(3)$ or full rigid body motions $SE(3)$. Exploiting the group structure, control laws that have been developed and applied to several systems such as quadrotors, UAVs and acquatic vehicles \citep{bullo1999tracking,lee2010geometric,fernando2011robust,rashadse3}.


While geometric control provides tools for stabilization and tracking, these objectives alone are not sufficient in many modern applications.
Robotic systems operating near obstacles, or under aggressive maneuvers must also ensure that the state remains within prescribed safety limits.  
Control Barrier Functions (CBFs) \citep{ames2016control,ames2019control,ferraguti2022safety} offer a principled way to address this need: they filter a nominal controller by enforcing state constraints in a minimally invasive manner, modifying the input only when a violation of the safe set becomes imminent.  
In the classical setting, CBFs are formulated for control–affine systems in euclidean spaces and enforced through min–norm quadratic programs (\cite{ames2016control,xu2015robustness}).  
Several variants have been proposed including exponential CBFs, adaptive and robust formulations, and higher-order extensions \citep{nguyen2016exponential,lopez}. 

The objective of generalizing CBFs to mechanical systems whose dynamics evolve on manifolds is present in much lower research effort.
%
%
An early approach is presented in \cite{wu2016safety}, which introduces a geometric construction for fully actuated mechanical systems by combining reciprocal barrier functions with geometric Control Lyapunov Functions. This method relies on reciprocal CBFs, which are not defined outside the safe set and therefore impose structural limitations on the resulting safety filter.
More recently, \cite{wang2025geometric} proposed a barrier function defined directly on the special Euclidean group $SE(3)$ by formulating the CBF conditions within a Riemannian framework. This approach enables multi-agent obstacle avoidance but the barrier functions depend only on configuration variables, without incorporation of velocity-dependent safety conditions. 
Extremely recently, in \cite{de2025bundles} the authors presented a new formulations for CBFs on Riemannian manifolds. In contrast to this work, we focus specifically on Lie groups and employ tools tailored to them.
%

The line of research conducted on euclidean spaces partly inspires this work. In particular, the approach presented in \cite{singletary2021safety} augments configuration-based barrier functions with kinetic-energy terms, thereby extending safety guarantees from the configuration level to the full dynamical setting.

As main contribution, we introduce and theoretically validate two classes of zeroing CBFs on Lie Groups. The first enforces kinematic constraints, suitable for implementing obstacle avoidance algorithms (the motivation for this class of CBFs mimics that in \cite{singletary2021safety,de2025bundles}, though their framework is not based on Lie groups). The second enforces kinetic energy limits along prescribed inertial-frame directions, relevant for ensuring safe physical interaction along rigid body motion tasks. This constitutes a new class of CBFs, both in technique and in motivation.

Numerical simulations validate both tasks of geometric obstacle avoidance and directional kinetic energy limitation in dynamically challenging scenarios.

The remainder of this document is structured as follows. Section \ref{sec:background} briefly reviews the geometric framework for mechanical systems evolving on a Lie groups and the theoretical foundations of CBFs. Section \ref{sec:cbf_lie_group} introduces and validates the proposed CBFs. Section \ref{sec:simulations} illustrates two distinct use cases of the proposed method through numerical simulations. Section \ref{sec:conclusions}
concludes the paper.


\section{Background} 

\subsubsection*{Notation.} We indicate with $G$ and $\mathfrak{g}$ respectively a finite-dimensional Lie group and its Lie algebra, identified with the tangent space at the identity $\mathfrak{g}=T_eG$, with $e \in G$ the identity group element. The dual vector space to $\mathfrak{g}$ is denoted $\mathfrak{g}^*$. For $g\in G$ and $\xi \in \mathfrak{g}$, $\text{Ad}_\xi: G \to G$ and $ \text{ad}_\xi: \mathfrak{g} \to \mathfrak{g}$ are respectively the group adjoint and algebra adjoint operators, and $\text{ad}^*_\xi: \mathfrak{g} \to \mathfrak{g}^*$ is the dual operator to $\text{ad}_\xi$. We consider matrix representations of $G$, i.e., an $n$-dimensional Lie group $G$ contains elements $g$ which are matrices (e.g, the identity matrix). The algebra element $\xi \in \mathfrak{g}\simeq \mathbb{R}^n$ admits a matrix representation $\xi ^{\wedge}$ of the same dimension of the matrix representation for $g\in G$, and we will make use of the canonical isomorphisms between the two representations. As relevant examples used in this work, for $\omega \in \mathfrak{so}(3)\simeq\mathbb{R} ^3$, $\omega^\wedge=-(\omega^{\wedge})^{\top}\in \mathbb{R}^{3 \times 3}$ indicates the skew-symmetric representation of the $3$-dimensional vector $\omega$.
Given a differentiable scalar function $h(x)$ we indicate with $\partial_x h$ its gradient covector, whose components are stored in a $n$-dimensional column, with $n$ the dimension of the manifold to which $x$ belongs. Given a matrix-valued field $w(x)\in \mathbb{R}^{n \times m}$, we indicate with $L_wh(x):=(\partial _xh)^{\top} w \in \mathbb{R}$ the Lie derivative of $h(x)$ in direction $w$. Finally $\textrm{skew}(\cdot)$ and $\textrm{Tr}(\cdot)$ produce the skew symmetric part and the trace of square matrices.

\label{sec:background}
\subsection{Mechanical Systems on Lie Groups}
The equations of motion for a fully actuated system evolving on an $n$-dimensional Lie group \(G\) are first introduced.
Let $\mathbb{I}:\mathfrak{g} \times \mathfrak{g}\to \mathbb{R}$ be an inertia tensor inducing a left-invariant metric on $G$, represented by a symmetric $\mathbb{R}^{n \times n}$ matrix. Let $u\in \mathfrak{g}^*\simeq \mathbb{R}^n$ be the body-fixed forces, $\eG\in G$ the configuration of the system and $\xi\in\mathfrak{g}\simeq \mathbb{R}^n$ the body-frame velocity. The evolution of the system can then be described by the following pair of Euler–Poincaré equations \citep{bullo1999tracking}:
\begin{equation}
    \label{eq:Lie group dynamics}
    \dot{\eG} =\eG \xi^{\wedge} \qquad \mathbb{I}\dot{\xi}= u+\text{ad}^*_\xi\mathbb{I}\xi
\end{equation}
where we recall that $g$ and $\xi^{\wedge}$ are represented by square matrices. The inertia tensor $\mathbb{I}$ characterizes the configuration-independent \textit{kinetic energy} of (\ref{eq:Lie group dynamics}), defined as
$E(\xi) = \frac{1}{2} \xi^{\top} \mathbb{I}\xi$. By introducing the state variable $x = (\eG,\xi)$, the dynamical equation \eqref{eq:Lie group dynamics} can be equivalently expressed in the control-affine form $\dot{x} = f(x) + \ug(x) u$
where:
\begin{equation}
    \label{eq: affine-control maps for Lie groups}
    f(x) = (\eG \xi^\wedge,\, \mathbb{I}^{-1} \text{ad}^*_\xi \mathbb{I}\xi), \qquad
    s(x) = (0,\, \mathbb{I}^{-1}).
\end{equation} 
%

\subsection{Control Barrier Functions }
Consider a control–affine system:
    \begin{equation}
        \label{eq:affine system}
        \dot x \;=\; f(x) + \ug(x)u, \;\;x\in D\subset\mathbb{R}^n,\; u\in U\subseteq\mathbb{R}^m
    \end{equation} 
    with $f(x),\ug(x)$ locally-Lipschitz functions.
    \\
\begin{definition} 
 A safe-set is given as the $0$–superlevel of a differentiable function $h:D \to \mathbb{R}$:
    \begin{equation}
        \label{eq: Safe-set def. with boundary and int}
        \begin{aligned}
            \mathcal{S} \;&=\; \{\,x\in D : h(x)\ge 0\,\}\\
            \partial\mathcal{S} \;&=\;\{\,x \in D : h(x)= 0\,\}\\
            \mathring{\mathcal{S}} \;&=\;\{\,x \in D: h(x)> 0\,\}
        \end{aligned}
    \end{equation}
where $\partial\mathcal{S}$ and $\mathring{\mathcal{S}}$ represent respectively the boundary and the interior set of $\mathcal{S}$.
\end{definition}

The goal of control barrier functions (CBFs) is to make the safe set $\mathcal{S}$ control invariant, i.e.,
\begin{equation}
    \forall x(0) \in \mathcal{S} \implies x(t)\in \mathcal{S} \,\,\, \forall t>0.
\end{equation}

In this work, we focus specifically on \textit{zeroing} CBFs, which are defined as follows: 
\\
\begin{definition} 
\label{def: CBF classical case}
    A smooth function $h(x):D\to\mathbb{R}$ with safe set $\mathcal{S}$ is a zeroing CBF on a domain $D^\prime$ with $\mathcal{S}\subseteq D^\prime\subseteq D$ if  $\partial_xh(x)\neq 0$ on $\partial \mathcal{S}$ and there exists an extended class-$\mathcal{K}$ function $\alpha$\footnote{A function $\alpha: (-b,a) \to (- \infty, \infty)$ with $a,b>0$, which is continuous, strictly increasing, and $\alpha(0)=0$.}such that, for all $x\in D^\prime$:
    \begin{equation}
        \sup_{u\in U}\;\bigl[L_f h(x) + L_\ug h(x)\,u \bigr] \;\ge\; -\,\alpha\!\bigl(h(x)\bigr).
        \label{eq:cbf-inequality}
    \end{equation}
\end{definition}

Zeroing CBFs differ from another widely studied class of CBFs, referred to as \textit{reciprocal} CBFs \citep{ames2016control}, which diverge at $\partial \mathcal{S}$ and have restricted applicability of operation within the interior of the safe set. Zeroing CBFs instead are well defined also outside of the safe set, i.e., on the whole set $D' \supseteq \mathcal{S}$. For zeroing CBFs the closed loop system, beyond invariance, achieves set stability for the safe set, as reported in the following result. 
\\
\begin{proposition}(\cite{ames2016control,xu2015robustness}).
    \label{prop: solution of a QP for safety}
    Let $h(x)$ be a zeroing CBF on $D'$ as in Definition \ref{def: CBF classical case}. Any Lipschitz continuous controller $u:D\to U$ satisfying \eqref{eq:cbf-inequality} renders $\mathcal{S}$ forward invariant. Furthermore, $\mathcal{S}$ is asymptotically stable on $D'$.
\end{proposition}
The way CBFs are used in practice is to employ them as safety filters that minimally modify a desired control input $u_\mathrm{des}(x,t)$ to ensure system safety, i.e., control invariance. This approach can be formalized as a Quadratic Program (QP):
\begin{equation*}
    \begin{aligned}
        u^\star(x,t) &\;=\; \arg\min_{u\in U}\;\|u-u_{\mathrm{des}}(x,t)\|^2\\
        &\text{s.t.}\;\; L_f h(x) + L_\ug h(x)\,u \;\ge\; -\,\alpha\!\bigl(h(x)\bigr),
    \end{aligned}
    \label{eq:cbf-qp}
\end{equation*}

where the feasibility of this QP is crucial for achieving control invariance of the safe set encoded by the candidate CBFs.
It is important to mention that candidate CBF must be of relative degree one, as otherwise the control input $u$ would not explicitly enter the constraint of the QP \citep{ames2019control}. This condition represents a restrictive assumption that is frequently not satisfied. In this work we will introduce classes of CBFs which structurally overcome this limitation.


\section{Safety Filters on Lie Groups} \label{sec:cbf_lie_group}

This section introduces a framework for constructing CBFs for mechanical systems evolving on a Lie group $G$ governed by the dynamics \eqref{eq:Lie group dynamics}. When a smooth function $h(x)$ depends only on the $g$ component of the state in \eqref{eq:Lie group dynamics}, we refer to it as a \textit{kinematic} function. As a matter of fact, kinematic functions are of relative degree $2$, and as such cannot be directly used as candidate CBFs.

Two specific scenarios will be analyzed in detail: (i) obstacle avoidance via configuration-dependent CBFs (a geometric adaptation of the strategy developed in \cite{singletary2021safety}) and (ii) directional energy regulation along a prescribed world-frame direction (a newly introduced class of zeroing CBFs).
\\
\begin{definition}
    \label{def:ECBF}
    Let $h(\eG): G\to\mathbb{R}$ be a kinematic smooth barrier function, and let $E(\xi)$ denote the kinetic energy of the system \eqref{eq:Lie group dynamics}. 
    For a design gain $\alpha_e>0$,  the energy-augmented barrier function is defined as:
    \begin{equation}
        H(x) := \alpha_e\, h(\eG)\;-\;E(\xi),
        \;\;\;\;x = (\eG,\xi)\in G\times\mathfrak{g}
        \label{eq:H-definition}
    \end{equation} 
    The associated energy-augmented safe set is:
    \begin{equation}
    \label{H safe set}
        \mathcal{S}_E
        :=\bigl\{x:H(x)\ge 0\bigr\}
        =\bigl\{x: E(\xi)\le \alpha_e\,h(\eG)\bigr\}.
    \end{equation}
\end{definition}
%

This formulation constitutes a direct use of the framework introduced by \cite{singletary2021safety} for robotic manipulators, here adapted for Lie groups. Using similar arguments as in that work, we demonstrate that the function defined in \eqref{eq:H-definition} indeed qualifies as a CBF. To this end, we first rigorously characterize the connection between the sets $\mathcal{S}_E$ and $\mathcal{S}$.
\\
\begin{proposition}
\label{prop:aug_safe_set}
    Let $h(\eG): G\rightarrow\mathbb{R}$ be a kinematic smooth barrier function with corresponding safe set $\mathcal{S}=\{g\in G: h(g)\geq 0 \}$ and  $H(\eG,\xi)=\alpha_eh(\eG)-E(\xi)$ be the associated energy-augmented barrier function \eqref{eq:H-definition} with corresponding safe set $\mathcal{S}_E$. Then,
    \begin{equation*}
        \text{(i) }\mathcal{S}_E \subset \mathcal{S}, \qquad\text{(ii) } \mathring{\mathcal{S}}\subset \lim_{\alpha_e\to\infty}\mathcal{S}_E\subset\mathcal{S}
    \end{equation*}
\end{proposition}
\begin{proof}
    (i) If $(\eG,\xi)\in\mathcal{S}_E$, then $H(\eG,\xi)=\alpha_eh(\eG)-E(\xi)\ge 0$, hence $\alpha_e h(\eG)\ge E(\xi)\ge 0$, which implies $\eG\in\mathcal{S}$. (ii) We can observe that $\mathcal{S}_E^i :=\{(\eG,\xi)\in G\times\mathfrak{g}: h(\eG)\geq \frac{1}{2}\frac{\lambda_\text{min}(\mathbb{I})}{\alpha_e^i}||\xi||^2\}\subset\mathcal{S}^{i+1}_E \subset \mathcal{S}$, where $\lambda$ is the minimum eigen-value of the matrix representation of the inertia tensor $\mathbb{I}$ and $\alpha_e^i<\alpha_e^{i+1}, i \in \mathbb{N}$. Since $E(\xi)/\alpha_e^i \to 0$ for $i \to \infty$ we have $\mathring{\mathcal{S}}\subset \lim_{i \to\infty}\mathcal{S}_E^i$.  
\end{proof}
From practical perspective, this result states that for a large enough $\alpha_e$ the safe set $\mathcal{S}_E$ (encoded by a CBF candidate with relative degree 1) can be used to achieve invariance of $\mathcal{S}$.
We next establish that the energy-augmented barrier function renders the set $\mathcal{S}_E$ forward invariant, consequently certifying \eqref{eq:H-definition} as a valid CBF.
\\
\begin{lemma}\label{lem:CBF} 
    Consider the system \eqref{eq:Lie group dynamics}, and a kinematic smooth barrier function $h(\eG):G\to\mathbb{R}$ with the associated safe-set $\mathcal{S}=\bigr\{\eG\in G:h(q)\geq0\bigr\}$ such that $\partial_g h(\eG)\neq 0$ on $\partial \mathcal{S}$. Let $H(\eG,\xi)$ be the energy-augmented barrier function \eqref{eq:H-definition} with the corresponding safe-set $\mathcal{S}_E$ \eqref{H safe set}. Then $H(\eG,\xi)$ is a CBF on $\mathcal{S}_E$ and, given a desired nominal controller $u_\text{des}(\eG,\xi)$, the controller: 
    \begin{equation*} 
        \label{eq: Lie Group QP} 
        \begin{aligned} 
            u^\star(\eG,\xi) = &\arg \min_{u}||u-u_\text{des}||^2\\ &\text{s.t.}\;\;\xi^Tu \leq \alpha_e\,L_fh(\eG) + \alpha(H(\eG,\xi))
        \end{aligned} 
    \end{equation*} 
 guarantees the forward invariance of $\mathcal{S}_E,\;\;\forall (\eG,\xi) \in \mathcal{S}_E$. 
\end{lemma}

\begin{proof} 
    To prove that \eqref{eq:H-definition} is a CBF on \eqref{H safe set}, according to Definition \ref{def: CBF classical case}, we start by noting that $H(\eG,\xi)$ is a smooth function on $G\times\mathfrak{g}$. Then, we rewrite the CBF constraint $L_fH + L_sH\,u \ge -\alpha(H)$ for \eqref{eq: affine-control maps for Lie groups}. In particular, $L_f H(\eG,\xi) = \alpha_e\,L_f h(\eG)$ since $ \xi^\top\operatorname{ad}^*_\xi(\mathbb{I}\xi) = 0$. Furthermore $L_s H(\eG,\xi) = \bigl(\partial_\xi H(\eG,\xi)\bigr)^\top\mathbb{I}^{-1} = (-\mathbb{I}\xi)^\top\mathbb{I}^{-1} = -\,\xi^\top$ given $\partial_g H(\eG,\xi) = \partial_g h(\eG)$, $\partial_\xi H(\eG,\xi) = -\partial_\xi E(\xi) = -\mathbb{I}\xi$. 
    Consequently, the constraint becomes:
    \begin{equation}
        \label{eq:ineq-CBF}
        \xi^\top u \;\le\; \alpha_eL_f h(\eG) + \alpha\bigl(H(\eG,\xi)\bigr),
    \end{equation}  
    which is the one reported in the lemma. Now we prove that $H(g,\xi)$ is a valid CBF.
    Notice that $\partial_xH$ is non-vanishing on $\partial\mathcal{S}_E$. In fact, we note that   $\partial_\xi H(\eG,\xi) \neq 0, \forall\xi\neq 0$ while $\partial_g H(\eG,0)=\partial_g h(\eG)\neq 0$  $\forall\xi=0$ because $\partial_g h(\eG)\neq 0$ on $\partial\mathcal{S}$ by assumption.
    Next, we prove that $\forall(\eG,\xi)\in\mathcal{S}_E, \exists u $ satisfying \eqref{eq:ineq-CBF}: 
    if $\xi\neq 0$, then $L_g H(\eG,\xi)=-\xi^\top\neq 0$ and the inequality \eqref{eq:ineq-CBF} is feasible as it defines a non-empty half-space in the input space; if $\xi=0$, then $\xi^\top u=0$ and 
    $L_f h(\eG) = 0$, since $\dot \eG = \eG\xi^\wedge = 0$. Furthermore, $\alpha(H(\eG,0))\ge 0, \forall(\eG,0)\in\mathcal{S}_E$, since $H(\eG,0)\ge 0$ and $\alpha$ is a class $\mathcal{K}$ function so the inequality holds $\forall u$. Notice that for (\ref{eq:Lie group dynamics}) we have $u\in U=\mathfrak{g}^*\simeq \mathbb{R}^n$, and as such the Lipschitz continuity of the controller resulting from the QP follows from \cite[Section 3]{xu2015robustness}.
    To conclude, the forward invariance of $\mathcal{S}_E$ under the closed–loop dynamics follows from the Proposition \ref{prop: solution of a QP for safety}.
\end{proof}

By combining Lemma \ref{lem:CBF} with Proposition \ref{prop:aug_safe_set}, it follows that the energy-augmented CBF also guarantees the forward invariance of the safe set $\mathcal{S}$, provided that the parameter $\alpha_e$ is chosen sufficiently large.



The kinematic function \(h(g)\) in \eqref{eq:H-definition} can be employed to represent the signed distance to a forbidden region, whereas the energetic augmentation term \(-E(\xi)\) guarantees that the system preserves a relative degree of one. One may also choose a constant $h(g)=E_{\max}/\alpha_e>0$ in \eqref{eq:H-definition} to enforce a desired bound $E_{\max}$ on the total kinetic energy. This choice makes the associated safe set $\mathcal{S}$ equal to the entire Lie group $G$, 
as on a Lie group the kinetic energy is independent on the configuration $g\in G$.
%
In contrast, when the dynamics are expressed on a general mechanical system in terms of generalized coordinates $q$ as in \cite{singletary2021safety}, the inertia metric depends on the configuration, preventing a clear decoupling between directional effects and inertial properties. 

We now introduce a new class of CBFs on Lie groups, representing directional kinetic energy constraints. These CBFs encode safe sets representing limits of kinetic energy in arbitrary directions of the ambient space on which the Lie Group acts. Although the definition (and the following proposition) holds for rigid-body motion ($G=SE(3)$), we present this section using the general formulation (\ref{eq:Lie group dynamics}) for technical compactness in the proofs and to facilitate straightforward specialization to $SO(3)$.
\\
%


\begin{definition} 
    Consider the system \eqref{eq:Lie group dynamics} with  inertia tensor $\mathbb{I}$. The \emph{directional kinetic energy along $n_v, n_\omega$} is defined as:
    \begin{equation}
        \label{eq:Ed_dir}
        E_{n_v, n_\omega}(\eG,\xi)
        :=
        \frac12\,\xi^\top P_{n_v, n_\omega}(\eG)\,\mathbb{I}\,\xi,
    \end{equation} 
    where $n_v, n_\omega \in \mathbb{S}^2$ denotes, respectively, the unit translational and rotational directions in the inertial frame, and $P_{n_v, n_\omega}(\eG)$ indicates a configuration-dependent projection matrix, defined as:
    $$P_{n_v, n_\omega}(\eG)
        :=
        \begin{bmatrix}
        n_{\omega B}(g) n_{\omega B}(g)^\top & \mathbf{0}\\[2pt]
        \mathbf{0} & n_{v B}(g) n_{v B}(g)^\top
        \end{bmatrix}
        \in \mathbb{GL}(6),$$  
    with $n_{v B}(g)$ and $n_{\omega B}(g)$, representing the corresponding unit directions in the body-fixed reference frame and $\mathbb{GL}(6)$ indicates the group of non singular $\mathbb{R}^{n \times n}$ matrices.
\end{definition}

$E_{n_v, n_\omega}(g,\xi)$ represents the sum of kinetic energies along the translational direction of the inertial frame $n_v$ and along the rotational axis defined by $n_\omega$, also in the inertial frame. We now construct a CBF that limits the kinetic energy along $n_v$ and $n_\omega$ to a maximum amount $E_{\max}>0$.
\\
\begin{proposition}  \label{prop:dir_energy_CBF}
    Given a desired energy bound $E_{\max}>0$ and a directional kinetic energy \eqref{eq:Ed_dir}, the function: 
    \begin{equation} \label{eq: H_dir} 
        H_{n_v, n_\omega}(\eG,\xi) := E_{\max}-E_{n_v, n_\omega}(\eG,\xi), 
    \end{equation} 
    with safe set: 
    \begin{equation}
    \label{eq: H_dir safeset}
        \mathcal{S}_{n_v, n_\omega} := \Big\{(\eG,\xi)\in G\times\mathfrak{g}: E_{n_v, n_\omega}(\eG,\xi)\le E_{\max}\Big\}
    \end{equation}
    is a (zeroing) CBF for (\ref{eq:Lie group dynamics}).
\end{proposition} 
\begin{proof}
    The constraint \eqref{eq:cbf-inequality} in this case is $\dot{H}_{n_v, n_\omega}+\alpha(H_{n_v, n_\omega})\geq0$, where:
    \begin{equation*}
        \begin{aligned}
            &\dot{H}_{n_v, n_\omega} = -\dot{E}_{n_v, n_\omega} 
            = -(\xi^TP_{n_v, n_\omega}\mathbb{I}\dot{\xi} +\frac{1}{2}\xi^T\dot{P}_{n_v, n_\omega}\mathbb{I}\xi)  \\
            &= -(\xi^TP_{n_v, n_\omega}(u+\operatorname{ad}^*_\xi\mathbb{I}\xi)+\frac{1}{2}\xi^T\dot{P}_{n_v, n_\omega}\mathbb{I}\xi)\\
            &= 
            -\xi^TP_{n_v, n_\omega}u - \xi^T(P_{n_v, n_\omega}\operatorname{ad}^*_\xi\mathbb{I}\xi+\frac{1}{2}\dot{P}_{n_v, n_\omega}\mathbb{I}\xi).
        \end{aligned}
    \end{equation*} 
    The term $\dot{P}_{n_v, n_\omega}$ can be expanded as:
    \[
    \begin{aligned}
       &\dot{P}_{n_v, n_\omega} 
       =  \begin{bmatrix}
           \dot{n}_{\omega B}\,n_{\omega B}^T+n_{\omega B}\,\dot{n}_{\omega B}^T\qquad\mathbf{0}\\
           \mathbf{0}\qquad \dot{n}_{vB}\,n_{vB}^T + n_{vB}\,\dot{n}_{vB}^T
       \end{bmatrix}\\ 
       &= \begin{bmatrix}
           -\omega^\wedge n_{\omega B}\,n_{\omega B}^T + n_{\omega B}\,n_{\omega B}^T\,\omega^\wedge \quad\mathbf{0}\\
           \mathbf{0}\quad -\omega^\wedge n_{vB}\,n_{vB}^T + n_{vB}\,n_{vB}^T \,\omega^\wedge
       \end{bmatrix}  \\
       &=  P_{n_v, n_\omega}\Omega - \Omega\,P_{n_v, n_\omega}   
        = [P_{n_v, n_\omega}(\eG),\,\Omega(\xi)] \\
       &=\operatorname{ad}_{P_{n_v, n_\omega}(\eG)}\Omega(\xi)
    \end{aligned}
    \]
    where we define the adjoint acting on the rotational part of the twist as the diagonal matrix $\Omega(\xi) := \text{diag}(\omega^\wedge, \omega^\wedge)$
    %
    and $\operatorname{ad}_{P_{n_v, n_\omega}(\eG)}:\mathfrak{gl}(6)\to\mathfrak{gl}(6)$ as the Lie bracket of the General Linear Lie group $\mathbb{GL}(6)$ (i.e., the classic matrix commutator).
    Therefore, we obtain:
    \begin{equation*}
        \begin{aligned}
            &\dot{H}_{n_v, n_\omega} = L_{s(x)}H_{n_v, n_\omega}\,u + L_{f(x)}H_{n_v, n_\omega} = \\
            &= -\xi^TP_{n_v, n_\omega}u - \frac{1}{2}\xi^T(2P_{n_v, n_\omega}\operatorname{ad}^*_\xi+\operatorname{ad}_{P_{n_v, n_\omega}}\Omega(\xi))\mathbb{I}\xi\\
            &= -\xi^TP_{n_v, n_\omega}u -\frac{1}{2}\xi^TI(\eG,\xi)\xi.
        \end{aligned}
    \end{equation*}
    where $I(\eG,\xi) := 2P_{n_v, n_\omega}\operatorname{ad}^*_\xi+\operatorname{ad}_{P_{n_v, n_\omega}}\Omega(\xi)$.
    Consequently, the CBF constraint can be written as:
    \begin{equation} \label{eq:world_dir_constraint}
        \xi^\top P_{n_v, n_\omega}(\eG) u \le \alpha(E_{\max}-\frac{1}{2}\xi^TP_{n_v, n_\omega}\mathbb{I}\xi) + \frac{1}{2}\xi^TI(\eG,\xi)\xi.
    \end{equation} 
    By applying the same arguments used in the proof of Lemma \ref{lem:CBF}, and noting that $\partial \mathcal{S} = \partial G=\partial (SE(3))=\varnothing$, the proof is complete. 
\end{proof}

Note that the virtual power term $\frac{1}{2}\xi^T I(\eG,\xi)\xi$ arises specifically from the presence of the time-varying projection matrix $P_{{n_v, n_\omega}}(\eG)$, which induces a nonzero drift term $L_{f(x)}H_{{n_v, n_\omega}}(\eG,\xi)$. This effect is a consequence of selecting a particular directional component of the world-frame and mapping it onto a body-fixed direction.

The Proposition \ref{prop:dir_energy_CBF} ensures that \eqref{eq: H_dir}, hereafter referred to as the \emph{directional energy CBF}, is able to impose an upper bound $E_{\max}$ on the kinetic energy of the closed-loop system along arbitrary world frame directions $n_v, n_\omega \in \mathbb{S}^2$.


\section{Simulations} \label{sec:simulations}

The proposed CBF has been tested in simulation on two specific control problems corresponding to the scenarios analyzed in the previous section: obstacle avoidance using configuration-dependent CBFs, and directional energy limiting along a fixed inertial axis.
%
%
Both simulations consider a rigid body in 3D space, evolving on the Lie group $SE(3)$ with pose $\eG = (R,p) \in SE(3)$, where $R \in SO(3)$ denotes the orientation and $p \in \mathbb{R}^3$ the position. The kinematic and dynamic equations (\ref{eq:Lie group dynamics}) specialize to:
\begin{equation*} \label{eq:SE3 dynamics}
    \begin{cases}
        \dot R = R\,\omega^\wedge,\\
        \dot p = R\,v,\\
        J\,\dot{\omega} = \omega \times J\omega  + f_{\omega},\\
        m\,\dot{v} = \omega \times m\,v + f_v,
    \end{cases}
\end{equation*}
where $\xi=(\omega, v) \in \mathfrak{se}(3)\simeq \mathfrak{so}(3) \times \mathbb{R}^{3}\simeq \mathbb{R}^3 \times \mathbb{R}^3$ is the body-frame \textit{twist}, with $\omega, v \in \mathbb{R}^3$ representing the angular and linear velocities \citep{lynch2017modern}, while $J$ and $m$ are the diagonal rotational inertial and the mass of the rigid body and $f_\omega , f_v \in \mathbb{R}^3$  denote external torque and force in the body frame.

As a nominal controller $u_\text{des}(R,p,\omega,v)=f_P+f_D+f_F$, we adopt the geometric tracking controller introduced in \cite{bullo1999tracking}, which incorporates a proportional–derivative feedback law together with a geometric feedforward term defined, respectively, as follows:\begin{equation*}
    \begin{aligned}
        &f_P = -\begin{bmatrix}
            \textrm{skew}(K_1R_d^TR)^{\vee}\\
            R^T(R+R_d)K_2(R^T+R_d^T)(p-p_d)
        \end{bmatrix}\\ 
        &f_D = -K_d\begin{bmatrix}
            \omega-R^TR_d\omega_d\\
            v - R^TR_d(v_d + \omega_d \times (R^T(p-p_d)))
        \end{bmatrix}\\ 
        &f_{F} = -\text{ad}^{*}_\xi\mathbb{I}Ad_{\eG^{-1} g_d}\xi_d + \mathbb{I}Ad_{\eG^{-1} g_d}\dot{\xi_d}
    \end{aligned} 
\end{equation*}

Here $g_d = (R_d, p_d)$ represents the reference trajectory on $SE(3)$, $\xi_d = (\omega_d, v_d)$ the corresponding reference twist.

%
%
%
In the $SE(3)$ case the two adjoints map $\text{Ad}_g$ and $\text{ad}_\xi$, respectively, on the group and on its algebra $\mathfrak{se}(3)$ admit the following matrix representations in $M_6(\mathbb{R})$:
\begin{equation*}\label{eq: adjoint representations in SE3}
    \operatorname{Ad}_g = 
            \begin{bmatrix}
                R & 0 \\
                p^\wedge R & R
            \end{bmatrix}
    \qquad
    \operatorname{ad}_\xi = 
            \begin{bmatrix}
                \omega^\wedge & 0 \\
                v^\wedge & \omega^\wedge
            \end{bmatrix} = (\operatorname{ad}_\xi^*)^T
\end{equation*}
In the simulation below, the rigid body consists of a disk with radius $r = 3~\mathrm{m}$, mass $m = 3~\mathrm{kg}$, and inertia $J$ = \text{diag}$(J_x, J_x, J_z)$ with $J_x = \frac{1}{4}mr^2$ and $J_z=\frac{1}{2}mr^2$ so that $\mathbb{I}=\textrm{diag}(J,mI_3)$.
Finally, the control gains are $K_1 = 20I_3$, $K_2 = 8I_3$, and $K_d =$ \text{diag}$(0.8I_3, 8I_3)$.

\subsection{Simulation 1: Collision-avoidance}
The first control problem demonstrates the ability of the proposed CBF to ensure safety in rigid-body dynamics during a constrained maneuver. Specifically, a disk-like rigid body must adjust its orientation to pass through multiple narrow slits in sequence with different orientations. 
Each slit is modeled as a pair of parallel planes with outward unit normals $\pm n$ and centers $c_L, c_R$, thereby defining a collision-free corridor of total width $d$. 
We define a function $h_o(R,p):SE(3)\to\mathbb{R}$ as a smooth minimum of the signed distances from the disk to the two planes of the slit:
\begin{equation*}
    h_o(R,p)
    := -\frac{1}{\beta}
      \log\!\left(e^{-\beta \psi_L}+e^{-\beta \psi_R}\right),
\end{equation*}
where $\beta>0$ is a sharpness parameter. 
The functions $\psi_L$ and $\psi_R$, defined as
$\psi_j(R,p) := n^\top(p-c_j) - s(R) - \delta, j=L,R$,  represent the signed distances from the disk to the left and right planes, including a small safety margin $\delta>0$.    
The term $s(R) := r\sqrt{1-(a^\top b)^2}$ measures the disk’s reach along the plane normal as a function of its orientation $R$, where $b$ is the body-fixed normal of the disk and $a = R^\top n$ is the slit’s normal expressed in the body frame. In this way, $s(R)$ acts as a support function along $n$, capturing the disk’s reach as its attitude varies.

The positivity of $h_o(R,p)$ guaranties that the disk does not collide with either plane of the slit. To avoid activating this constraint when the disk is far from the slit, we introduce a smooth weighting function $\chi:\mathbb{R}^3\to\mathbb{R}$ defined as $\chi(p) := \exp\left(-\frac{1}{2} (p-c)^\top \Sigma^{-1} (p-c)\right)$ with $c := \frac{c_L + c_R}{2} + \Delta$. 
The smooth barrier function is then constructed as
\begin{equation*}
    \label{eq: Sim1 def h}
    h(R,p) := (1-\chi(p)) K + \chi(p) h_o(R,p)
\end{equation*}
where $\Sigma = \sigma^2I_3\succ0$ is a matrix defining the obstacle as a sphere of radius $\sigma$, $\Delta$ is a pure translation parameter used a degree of freedom to decide the conservativeness of $h$ and $K>0$ is a constant parameter defining the upper bound of the CBF.

Adding the energetic term as in \eqref{eq:H-definition}, we obtain the energy-augmented CBF:
\begin{equation*} \label{eq: H sim1}
    H(R,p,\xi) = \alpha_eh(R,p)-\frac{1}{2}\xi^T\mathbb{I}\xi
\end{equation*}
Since the environment contains two slits, we define an energy-augmented CBF for each one and enforce safety by solving the following QP with two constraints:
\begin{equation*} \label{eq: QP sim1}
    \begin{aligned}
        u^\star = \arg\min_{u\in\mathbb{R}^{6}} \ &\|u - u_{\mathrm{des}}\|^{2} \\
        \text{s.t.}\quad 
        &\xi^\top u \leq L_f h_1(R,p) + \alpha\!\left(H_1(R,p,\xi)\right)\\
        &\xi^\top u \leq L_f h_2(R,p) + \alpha\!\left(H_2(R,p,\xi)\right)
    \end{aligned}
\end{equation*}
where $h_1,h_2$ and $H_1,H_2$ denote the kinematic and energy-augmented CBF associated with the two slits.

The simulation runs for $15\mathrm{s}$ with slits centered at $c_1 =[2.8, 1.0, 1.6]$ and $c_2 =[2.8, -2.0, 1.6]$, defining a gap of $d = 0.3 \mathrm{m}$. The first slit is vertical, while the second is rotated $45^\circ$ about the $Y$–axis, as shown in Figure \ref{fig:cbf-slit}. The CBF $h$ is parameterized by $\beta=25, \sigma=12, \Delta=\tfrac{1}{2}\mathbf{e}_2, K=\tfrac{\alpha_e}{2}$, while the QP constraint employs a linear class-$\mathcal{K}$ function with coefficient $\alpha=1.0$.

Figure~\ref{fig:cbf-slit} shows a visual representation of the simulation setup, while the temporal evolution of the geometric terms $\alpha_e h_i$ and the energy-augmented CBFs $H_i$ for each slit $i\in\{1,2\}$ is depicted in Figure~\ref{fig:cbf-slit2} for two different values of $\alpha_e>0$. Consistent with the theoretical results, both quantities remain nonnegative over time, thereby enabling the disc to traverse the slits without colliding with the surrounding walls. Safety is ensured for both values of $\alpha_e$. However, larger values of $\alpha_e$ increase the repulsive effect exerted by the boundary of the safe set, leading the controller to preserve a greater margin from the obstacles. In narrow passages, where obstacles are present on both sides, this configuration becomes more susceptible to oscillatory behavior, as the system must continuously adjust its state to comply with the more restrictive safe region.
%



\begin{figure}[t]
    \centering 
    \includegraphics[width=0.8\linewidth]{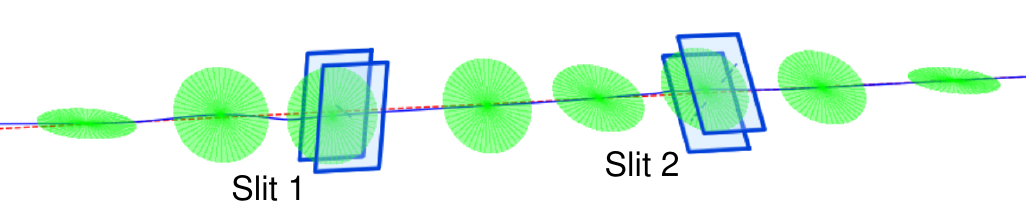} 
    \vspace{-1mm}
    \caption{3D visualization of the obstacle-avoidance simulation employing the propsed CBF with $\alpha_e = 150$.
    Green disks depict the time-varying pose of the rigid body as it tracks the (red dashed) reference trajectory, while the safety filter adjusts  the orientation as required to fulfill the obstacle avoidance objective.    }
    \label{fig:cbf-slit}
\end{figure}

\begin{figure}[t]
    \centering
    \includegraphics[width=\linewidth]{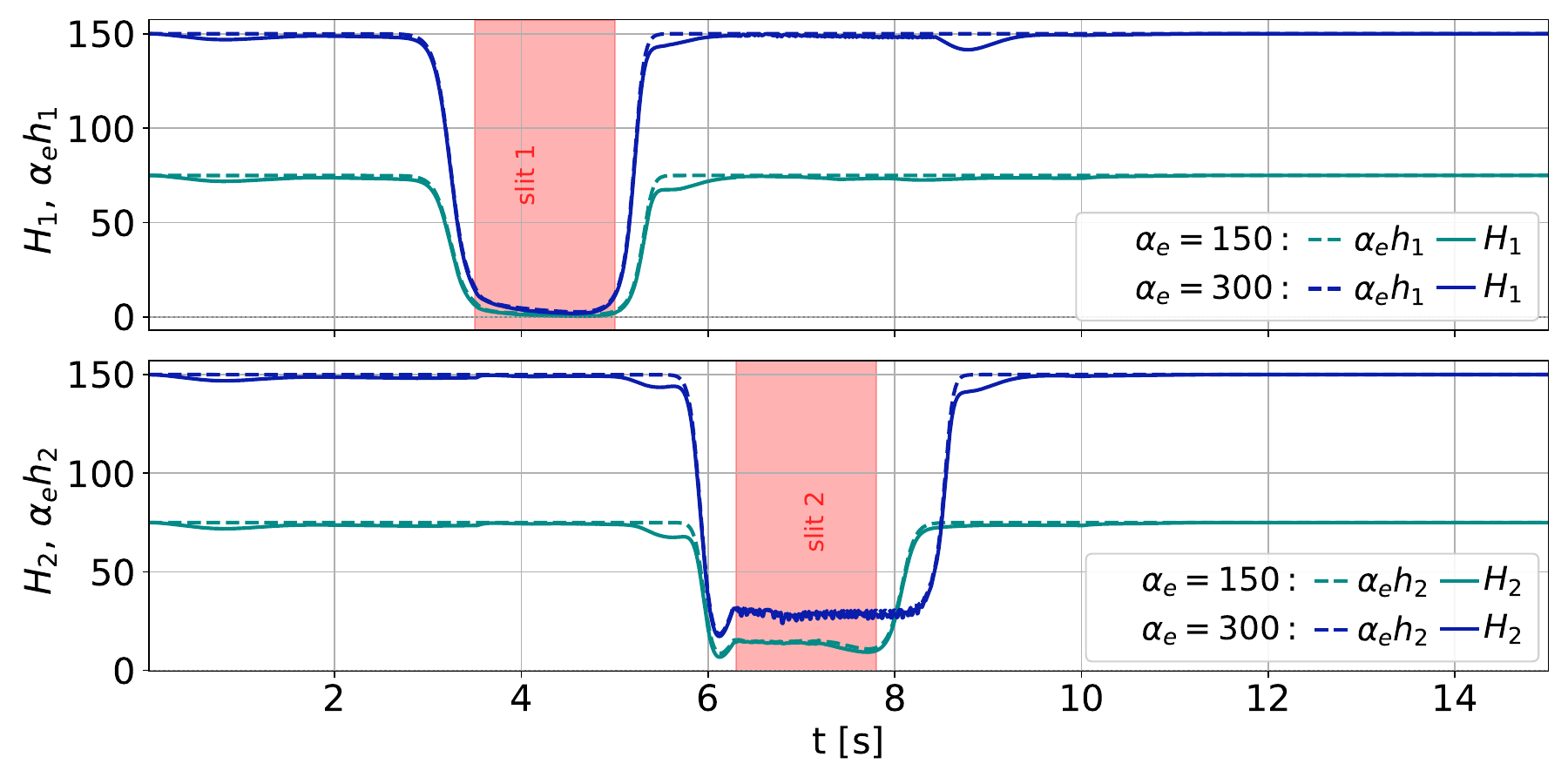}\vspace{-1mm}
    \caption{Temporal evolution of $\alpha_e h_i$ and $H_i$ associated with the $i$-th slit, where $i \in \{1,2\}$.
    Red shaded regions in Figure~\ref{fig:cbf-slit} mark the time intervals during which the disk is within each slit.}  
    \label{fig:cbf-slit2}
\end{figure}

\subsection{Simulation 2: Limiting directional energy}

The second control objective is to impose constraints on the component of kinetic energy projected along a prescribed direction. In practical implementations, this objective can be exploited to attenuate or prevent damage that may arise during physical interactions between the rigid body (e.g., a fully actuated unmanned aerial vehicle) and the environment (e.g., the ground).  
The system starts from $p(0) = [15,0,10]^\top$ and $R(0) = R_x\left(\frac{\pi}{2}\right)$, an undesirable $90^\circ$ tilt about the $x$-axis. The disk must execute a vertical landing on a pad located at the world-frame origin, lying on a plane with normal $n_v = \mathbf{e}_1 \times \mathbf{e}_2$, where $\mathbf{e}_i$ is the $i$-th canonical basis vector in $\mathbb{R}^3$. Within this setup, the goal is to enforce kinetic-energy limits along the pad normal $n_v$ to maintain safe interaction dynamics.
Let $n_{v B} = R^\top n_v$ be the pad normal expressed in the body frame.  
The linear velocity along this direction is $v_n = n_{v B}^\top v$.
The directional kinetic energy is then $E_n= \tfrac12 mv_n^2$. 
The projector is the diagonal matrix $P_{n_{v}}  = \text{diag}(\mathbf{0},n_{v B}n_{v B}^T)\in\mathbb{R}^{6\times 6}$, where, for brevity, the dependence on $R$ has been omitted.
To enforce the desired energy bound \(E_n \le E_{n,\max} = 1.5J\), we employ the CBF introduced in Proposition \ref{prop:dir_energy_CBF}. This is achieved by solving the QP associated with the affine inequality constraint  \eqref{eq:world_dir_constraint}, which can be reformulated as:
\begin{equation*}
    \label{eq: sim3 constraint}
    \begin{aligned}
    (\partial_\xi E_n)^{\top}&I^{-1} u
    \;\le\;
    \alpha\!\left(E_{n,\max} - E_n\right)\, +\\
    & + \frac{1}{2}\xi^T\Biggl(
    \begin{bmatrix}
        \mathbf{0}\qquad\qquad\qquad\qquad\quad\mathbf{0}\\
        \;\,\mathbf{0}\;\;3n_Bn_B^T\,\omega^\wedge-\omega^\wedge n_Bn_B^T
    \end{bmatrix}\Biggr)I\xi
    \end{aligned}
\end{equation*}

\begin{figure}[t]
    \centering
    \includegraphics[width=\linewidth]{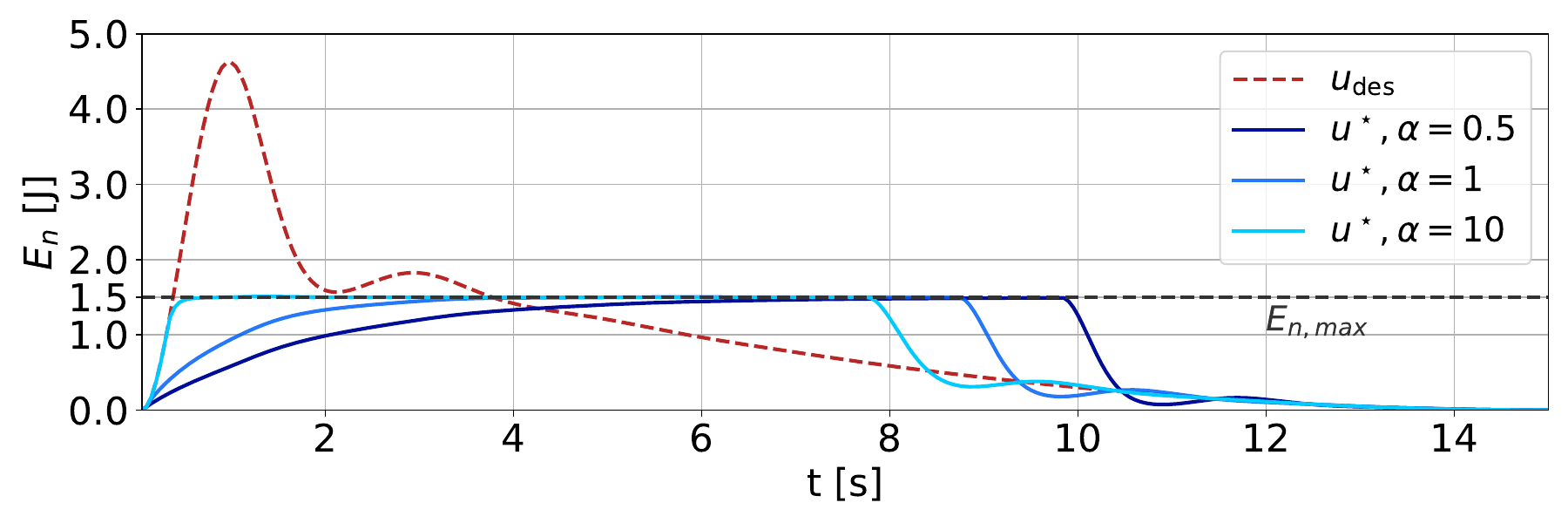}\vspace{-1mm}
    \caption{Comparison of the kinetic energy along the direction normal to the landing pad obtained using the nominal controller \(u_{\text{des}}\) (red dashed line) and the corresponding kinetic energy obtained using the proposed controller for different values of \(\alpha  \).}
    \label{fig:sim2_energy}
\end{figure}
The results in Figure \ref{fig:sim2_energy} show that the nominal controller drives the vehicle toward the pad with excessive kinetic energy, overshooting the admissible normal-direction energy and producing an unsafe contact. In contrast, the proposed CBF modulates the motion towards the pad limiting the kinetic energy to $E_{n,\max} = 1.5J$.

\section{Conclusions} \label{sec:conclusions}

In this work, we presented two classes of energy-augmented control barrier functions  enabling safety-critical control for rigid-body systems evolving on Lie groups.
The proposed formulation supports both configuration-dependent obstacle avoidance and world-frame directional energy limiting through a single analytical construction. 
Future work will explore other energetic applications such as energy limiting in body-frame directions, multi-agent interactions and hardware deployment.

\bibliography{ifacconf}

\end{document}